\documentclass[journal,10pt,onecolumn,draftcls]{IEEEtran}
\IEEEoverridecommandlockouts

\usepackage{amsmath, amsthm, amssymb, amsfonts, mathtools, xargs, tensor, units, cite, stmaryrd, mathrsfs, algpseudocode, comment, graphicx}
\usepackage{mathrsfs}
\usepackage{multibib}
\usepackage[svgnames]{xcolor}
\usepackage[unicode=true, bookmarks=true, bookmarksnumbered=false, bookmarksopen=false, breaklinks=false, pdfborder={0 0 1}, backref=false, colorlinks=false, hidelinks]{hyperref}
\makeatletter
\let\NAT@parse\undefined
\makeatother

\newtheorem{manuallemmainner}{Lemma}
\newenvironment{manuallemma}[1]{%
  \renewcommand\themanuallemmainner{#1}%
  \manuallemmainner
}{\endmanuallemmainner}
\newtheorem{manualtheoreminner}{Theorem}
\newenvironment{manualtheorem}[1]{%
  \renewcommand\themanualtheoreminner{#1}%
  \manualtheoreminner
}{\endmanualtheoreminner}
\theoremstyle{definition}
\newtheorem{defn}{Definition}

\theoremstyle{definition}
\newtheorem{assum}{Assumption}
\newtheorem{lem}{Lemma}
\newtheorem{thm}{Theorem}

\theoremstyle{remark}

\allowdisplaybreaks
	\title{\Large Safe Controller for Output Feedback Linear Systems using Model-Based Reinforcement Learning}
\author{\normalsize S M Nahid Mahmud$^{1}$ \and Moad Abudia$^{1}$ \and Scott A Nivison$^{2}$ \and Zachary I. Bell$^{2}$ \and Rushikesh Kamalapurkar$^{1}$
\thanks{*This research was supported, in part, by the Air Force Research Laboratories under award number FA8651-19-2-0009. Any opinions, findings, or recommendations in this article are those of the author(s), and do not necessarily reflect the views of the sponsoring agencies.}%
\thanks{$^{1}$School of Mechanical and Aerospace Engineering, Oklahoma State University, email: {\tt\footnotesize \{nahid.mahmud, abudia@okstate.edu, rushikesh.kamalapurkar\} @okstate.edu}.}%
\thanks{$^{2}$ Air Force Research Laboratories, Florida, USA, email: {
\tt\footnotesize \{scott.nivison, zachary.bell.10\}
 @us.af.mil.}}}
\begin{document}
\maketitle
\thispagestyle{plain}
\pagestyle{plain}
\begin{abstract} 
\normalsize  The objective of this research is to enable safety-critical systems to simultaneously learn and execute optimal control policies in a safe manner to achieve complex autonomy. Learning optimal policies via trial and error, i.e., traditional reinforcement learning, is difficult to implement in safety-critical systems, particularly when task restarts are unavailable. Safe model-based reinforcement learning techniques based on a barrier transformation have recently been developed to address this problem. However, these methods rely on full state feedback, limiting their usability in a real-world environment. In this work, an output-feedback safe model-based reinforcement learning technique based on a novel barrier-aware dynamic state estimator has been designed to address this issue. The developed approach facilitates simultaneous learning and execution of safe control policies for safety-critical linear systems. Simulation results indicate that barrier transformation is an effective approach to achieve online reinforcement learning in safety-critical systems using output feedback.
\end{abstract}

\section{Introduction}

Over the past decade, the topic of safe reinforcement learning has gained a lot of attention in the disciplines of robotics and controls. One of the primary reasons for this is the increase in the expectation of autonomy in safety-critical systems in the real-world tasks. While unmanned autonomous systems have significant advantages such as repeatability, precision, and lack of physical weariness over their non-autonomous and biological counterparts, they are often costly to construct and restore. To avoid failures during the learning phase, methods that allow the unmanned autonomous agents to learn to perform tasks with safety guarantees are needed.

In past, Reinforcement learning (RL) has been demonstrated to be an effective approach for synthesizing online optimal policies for both known and unknown discrete/continuous-time dynamical systems \cite{SCC.Sutton.Barto.ea1992,SCC.Doya2000}. However, due to sample inefficiencies, RL often necessitates a large number of iterations. Model-based reinforcement learning (MBRL) techniques can enhance sample efficiency in RL\cite{SCC.Wawrzynski2009,SCC.Zhang.Cui.ea2011,SCC.Kamalapurkar.Walters.ea2016a,SCC.Kamalapurkar.Rosenfeld.ea2016,SCC.Kamalapurkar.Walters.ea2016}. Generally MBRL techniques guarantee stability, not safety. In recent years, significant progress has been made in developing Safe Model-based Reinforcement learning (SMBRL) techniques to learn safe controllers for different classes of  systems\cite{SCC.Howard1960,SCC.Meyn.Tweedie1992,SCC.Puterman2014,SCC.Fisac.Akametalu.ea2018,SCC.Li.Kalabi.ea2018,SCC.Yutong.Nan.ea2021,SCC.Cohen.ea2020,SCC.Yang.Vamvoudakis.ea2019,SCC.Greene.Deptula.ea2020,SCC.Mahmud.Hareland.ea2021,SCC.Mahmud.Nivison.ea2021,SCC.Mahmud.Moad.ea2021}. While Markov decision process (MDP) based SMBRL methods have been available for discrete time systems with finite state and action spaces \cite{SCC.Howard1960,SCC.Meyn.Tweedie1992,SCC.Puterman2014}, synthesizing online controllers for systems in continuous time, under output feedback, while guaranteeing stability and safety is still a challenging problem.

SMBRL techniques that provide provide probabilistic safety guarantees for continuous-time (CT) stochastic systems have recently been studied in results such as \cite{SCC.Fisac.Akametalu.ea2018,SCC.Li.Kalabi.ea2018,SCC.Yutong.Nan.ea2021}; however, not all applications are conducive to probabilistic safety guarantees. Applications such as manned aviation demand deterministic safety guarantees, and as such, online, real-time learning in such systems is challenging. Recently, SMBRL techniques for CT deterministic systems have been studied in results such as \cite{SCC.Cohen.ea2020,SCC.Yang.Vamvoudakis.ea2019,SCC.Greene.Deptula.ea2020, SCC.Mahmud.Hareland.ea2021,SCC.Mahmud.Nivison.ea2021, SCC.Mahmud.Moad.ea2021}. In  \cite{SCC.Cohen.ea2020}, a SMBRL method has been developed to synthesis a real-time safe controller online by incorporating proximity penalty method developed in \cite{SCC.Walters.Kamalapurkar.ea2015,SCC.Deptula.ea2020} with the framework of control barrier functions. While the control barrier function results in safety guarantees, the existence of a smooth value function, in spite of a nonsmooth cost function, needs to be assumed \cite{SCC.Mahmud.Nivison.ea2021}. 

This paper is inspired by the nonlinear coordinate transformation first introduced in \cite{SCC.Graichen.Petit.ea2009}. Leveraging the results of \cite{SCC.Graichen.Petit.ea2009}, a barrier transformation (BT) to construct an equivalent, unconstrained optimal control problem from a state constrained optimal control problem was introduced in \cite{SCC.Yang.Vamvoudakis.ea2019}. The unconstrained problem is then solved using an adaptive optimal control method under persistence of excitation (PE). To soften the restrictive PE requirement, \cite{SCC.Greene.Deptula.ea2020} utilized a MBRL formulation to yield a SMBRL technique to synthesize safe controllers. However, the SMBRL method in \cite{SCC.Greene.Deptula.ea2020} requires exact model knowledge; to address this limitation, \cite{SCC.Mahmud.Nivison.ea2021} extended the results in \cite{SCC.Greene.Deptula.ea2020} to yield a SMBRL solution to the online state-constrained optimal feedback control problem under parametric uncertainty using a filtered concurrent learning technique.

While results such as \cite{SCC.Yang.Vamvoudakis.ea2019, SCC.Greene.Deptula.ea2020, SCC.Mahmud.Hareland.ea2021, SCC.Mahmud.Nivison.ea2021} provide verified safe feedback controllers, they all rely on full state feedback. To achieve safe learning using output feedback, \cite{SCC.Mahmud.Moad.ea2021} extended the results in \cite{SCC.Yang.Vamvoudakis.ea2019, SCC.Greene.Deptula.ea2020,SCC.Mahmud.Hareland.ea2021,SCC.Mahmud.Nivison.ea2021} for nonlinear control-affine systems in Brunovsky canonical form, where the state comprises of the output and its derivatives. This paper focuses on extending the results (OF-SMBRL) in \cite{SCC.Mahmud.Moad.ea2021} by developing a SMBRL technique for partial observable state-constrained linear systems, not necessarily in the Brunovski canonical form, with a more general output equation.

The primary challenge in output feedback SMBRL is that the BT preservs neither the linearity, nor the Brunovsky canonical form of the system. As such, observer development in the transformed coordinates is difficult. While state estimation can be done in the original coordinates, due to the nature of the BT, small estimation errors in the original coordinates do not translate to small errors in the transformed coordinates. Since the controllers are designed in transformed coordinates, large state estimation errors in the transformed coordinates can yield unexpected results. In this paper, a novel Luenberger-like \cite{SCC.Luenberger1964} BT based adaptive observer is designed using Lyapunov methods, and is integrated in a BT-based SMBRL framework, to learn feedback control policies with guaranteed safety and stability during the learning and execution phases.

In the following, Section \ref{control object} formalize problem statement. \ref{Barrier Transformation} introduces the BT. Section \ref{sec:Velocity-estimator-designCH3} and Section \ref{Stability Analysis for the state estimation} detail the novel state estimator and an analysis of bounds on the resulting state estimation errors, respectively. Section \ref{Model-Based Reinforcement Learning} describes the novel SMBRL technique for synthesizing feedback control policies in transformed coordinates. In Section \ref{sec:Analysis}, a Lypaunov-based analysis, in the transformed coordinates, is utilized to establish practical stability of the closed-loop system resulting from the developed SMBRL technique. Guarantees that the safety requirements are satisfied in the original coordinates are also established. Simulation results in Section \ref{Simulation} show the performance of the developed SMBRL approach.  

\section{Problem Formulation} \label{control object}
We consider the following continuous-time linear dynamical system. 
\begin{align}\label{eq:Dynamics}
\dot{x} = Ax+Bu , \quad y = Cx 
\end{align}
where $x \coloneqq [x_{1};\hdots;x_{n}] \in \mathbb{R}^{n}$ is the system state, $A \in \mathbb{R}^{n \times n}$ is the transition matrix, $B \in \mathbb{R}^{n \times m}$ is the control effectiveness matrix, $u \in \mathbb{R}^{m}$ is the control input, $C \in \mathbb{R}^{q \times n}$ is output matrix, and $y \in \mathbb{R}^{q}$ is the measured output. To make the barrier transformation feasible, the following structure is imposed on the problem.
\begin{assum}\label{Assumption-stateestimatorgain}
    The output is comprised of a selection of state variables. That is, every row of the matrix $C$ has exactly one element equal to one, and the rest of the elements are zero.
\end{assum}

Let, $\hat{x} \coloneqq [\hat{x}_{1};\hdots;\hat{x}_{n}] \in \mathbb{R}^{n}$ be the estimates of $x$, where the notation $[v;w]$ is used to denote the vector $[v^T \quad w^T]^{T}$. The objective is to design an adaptive estimator to estimate the state variables online, using input-output measurements, and to simultaneously estimate and utilize an optimal feedback controller, $u$ such that starting from a given estimated feasible initial condition $x^{0}$, the trajectories $x(\cdot)$ decay to the origin and satisfy $x_{i}(t) \in (\underline{z}_{i},\overline{z}_{i}), \forall i = 1,\hdots,n$ for some constants $ \underline{z}_{i}<0<\overline{z}_{i}$. The notation $(\cdot)_{i}$ is used in the rest of the manuscript to denote the $i$th element of the vector $(\cdot)$, and the notation $I_{o}$ denotes the identity matrix of size $o$.



The following section introduces the barrier transformation and formalizes the connection between trajectories of the transformed system and the original system.

\section{Barrier Transformation}\label{Barrier Transformation}

\begin{defn}\label{def1}
Given $\underline{\mathscr{E}}<0<\overline{\mathscr{E}}$, The function $b_{(\underline{\mathscr{E}},\overline{\mathscr{E}})} : \mathbb{R} \rightarrow \mathbb{R}$, referred to as the barrier function (BF), is defined as $ b_{(\underline{\mathscr{E}},\overline{\mathscr{E}})}(\mathscr{Y}) \coloneqq \log \left(\frac{\overline{\mathscr{E}}(\underline{\mathscr{E}} - \mathscr{Y})}{\underline{\mathscr{E}}(\overline{\mathscr{E}} - \mathscr{Y})}\right)$.
\end{defn}

The inverse of the BF on the interval $(\underline{\mathscr{E}},\overline{\mathscr{E}})$ containing the origin is given by $ b^{-1}_{(\underline{\mathscr{E}},\overline{\mathscr{E}})}(\mathscr{Y}) = \left(\underline{\mathscr{E}}\overline{\mathscr{E}}\frac{e^{\mathscr{Y}} - 1}{\underline{\mathscr{E}}e^{\mathscr{Y}} - \overline{\mathscr{E}}}\right) $. In this paper, the barrier transformation is a nonlinear coordinate transformation, defined as $s_{i} \coloneqq b_{(\underline{z}_{i},\overline{z}_{i})}(x_{i}), \quad \text{and} \quad x_{i} = b^{-1}_{(\underline{z}_{i},\overline{z}_{i})}(s_{i}),$ where $s_{i}$ is the $i$th element of the transformed system state vector, $s$. 
Evaluating the derivative of the inverse of the barrier function at $s_{i}$ yields $\frac{\mathrm{d}b^{-1}_{(\underline{z}_{i},\overline{z}_{i})}(s_{i})}{\mathrm{d}s_{i}} = \frac{1}{T_{i}(s_{i})}$, where $T_{i}(s_{i}) \coloneqq \frac{\underline{z}_{i}^{2} e^{s_{i}} - 2\underline{z}_{i} \overline{z}_{i} + \overline{z}_{i}^{2} e^{-s_{i}}}{\overline{z}_{i}\underline{z}_{i}^{2} - \underline{z}_{i} \overline{z}_{i}^{2}}.$ Let $T(s) \coloneqq [T_{1}(s_{1}); \hdots ; T_{n}(s_{n})]$.

Let $\underline{z} = [\underline{z}_{1}; \hdots ;\underline{z}_{n}]$, and $\overline{z} = [\overline{z}_{1}; \hdots ;\overline{z}_{n}]$. In the following, for any vector $\mathscr{L}$, comprised of components of $x$, such that $\mathscr{L} = [(x)_p;\hdots;(x)_q]$, with $1\leq p \leq q \leq n$, the notation $b(\mathscr{L})$ is used to denote componentwise application of the barrier function with the appropriate limits selected componentwise from the vectors $\underline{z}$ and $\overline{z}$. That is, $b(\mathscr{L}) \coloneqq [b_{((\underline{z})_p,(\overline{z})_p)}((x)_p); \hdots ; b_{((\underline{z})_q,(\overline{z})_q)}((x)_q)]$. Similarly, given any vector $\mathscr{L}$, comprised of components of $s$, such that $\mathscr{L} = [(s)_p,\cdots,(s)_q]^T$, with $1\leq p \leq q \leq n$, $b^{-1}(\mathscr{L}) \coloneqq [b^{-1}_{((\underline{z})_p,(\overline{z})_p)}((s)_p); \hdots ; b^{-1}_{((\underline{z})_q,(\overline{z})_q)}((s)_q)]^T$.

To transform the dynamics in \eqref{eq:Dynamics} using the BT, the time derivative of the transformed state, $s \in \mathbb{R}^{n}$, can be computed as 
\begin{align}\label{eq:BTDynamics}
\dot{s} = T(s)\odot
(Ax+Bu)= F(s) + G(s)u
\end{align}
where, $\odot$ represents Hadamard product \cite{SCC.Horn.ea1990}, $F(s) \coloneqq T(s)\odot Ab^{-1}(s)$, $G(s) \coloneqq T(s)\odot B$. 


\subsection{Analysis of trajectories}
In the following lemma, the trajectories of the original system and the transformed system are shown to be related by the barrier transformation provided the trajectories of the transformed system are \emph{complete}\cite{SCC.Mahmud.Nivison.ea2021}.

\begin{lem}\label{lem:trajectoryRelation}If $t \mapsto \Phi\big(t,b(x^{0}),\zeta\big)$ is a complete Carath\'{e}odory solution to \eqref{eq:BTDynamics}, starting from the initial condition $b(x^{0})$, under the feedback policy $(s,t) \mapsto \zeta (s,t)$ and $t \mapsto \Lambda(t,x^{0},\xi)$ is a Carath\'{e}odory solution to \eqref{eq:Dynamics}, starting from the initial condition $x^{0}$, under the feedback policy $(x,t) \mapsto \xi(x,t)$, defined as $\xi(x,t) = \zeta(b(x),t)$, then $t \mapsto \Lambda(t,x^{0}, \xi)$ is complete and $ \Lambda(t,x^{0}, \xi) = b^{-1}\left(\Phi(t,b(x^{0}),\zeta)\right) $ for all $t \in \mathbb{R}_{\geq 0}$.
\end{lem}
\begin{proof}
See \cite[Lemma 1]{SCC.Mahmud.Nivison.ea2021}.
\end{proof}

It can be noted that the feedback $\xi$ is well-defined at $x$ only if $b(x)$ is well-defined, which is the case whenever $x$ is inside the barrier. As such, the main conclusion of the lemma also implies that $\Lambda(\cdot,x^0,\xi)$ remains inside the barrier. It is thus inferred from Lemma \ref{lem:trajectoryRelation} that if the trajectories of \eqref{eq:BTDynamics} are bounded and decay to a neighborhood of the origin under a feedback policy $(s,t) \mapsto \zeta (s,t)$, then the feedback policy $(x,t) \mapsto \zeta \big(b(x),t \big)$, when applied to the original system in \eqref{eq:Dynamics}, achieves the control objective stated in Section \ref{control object}.

In the following section, a state estimator, that allows for rigorous inclusion of state estimation errors in the analysis of the controller, is developed. 

\section{State Estimator}\label{sec:Velocity-estimator-designCH3}

In this section, a barrier based adaptive state estimator inspired by Luenberger observer \cite{SCC.Luenberger1964} has been designed to generate estimates of $x$. Furthermore, the design is motivated by the need to obtain the error bound in Lemma \ref{lem3}, and the designed estimator is given by
\begin{align}\label{eq:observerDynamics2}
\dot{\hat{x}} = [\dot{\hat{x}}_{1};\cdots;\dot{\hat{x}}_{n}]
\end{align} with 
	$\dot{\hat{x}}_{i}  \coloneqq (Bu)_{i}+\frac{1}{T_i(b(\hat{x}_i))} \bigg(Ab(\hat{x})+L\left(y_{m}-Cb(\hat{x})\right)\bigg)_{i}$,
where, $i = [1;\cdots;n]$, $y_{m}$ = $Cb(x)$, and $L \in \mathbb{R}^{n \times q}$ is the gain matrix.


To transform the dynamics in \eqref{eq:observerDynamics2} using the BT, the time derivative of the estimated transformed state, $\hat{s} \in \mathbb{R}^{n}$, can be computed as 
\begin{align}
	\dot{\hat{s}}  &= T(\hat{s})\odot(B\hat{u})+A\hat{s}+L\left(y_{m}-Cb(\hat{x})\right) \nonumber \\&= A\hat{s}+G(\hat{s})u+LC\tilde{s}.
	\label{eq:observerdynamics_s1}
\end{align}

Note that the relationship  $ b(Cx) = Cb(x) $, leveraged in the computation above, is only true under Assumption \ref{Assumption-stateestimatorgain}.

In the transformed coordinates, the estimator error can be computed as 
\begin{align}\label{eq:BTobserverDynamics}
	\dot{\tilde{s}} = F(s)+G(s)u-G(\hat{s})\hat{u} -As+A\tilde{s}-LC\tilde{s},
\end{align}

with $	\tilde{s} = s -\hat{s}, \quad \dot{\tilde{s}} = \dot{s} - \dot{\hat{s}}$.

As detailed in Lemma \ref{lem2} below, the design of the BT ensures that the trajectories of \eqref{eq:Dynamics}, \eqref{eq:BTDynamics}, \eqref{eq:observerDynamics2}, and \eqref{eq:BTobserverDynamics} are linked by the BT whenever the underlying state trajectories $x(\cdot)$ and $s(\cdot)$ and the initial conditions $\hat{x}^{0}$ and $\hat{s}^{0}$ are linked by the BT. 

\begin{lem}\label{lem2}If $t \mapsto \Psi\big(t;b(x(\cdot)),b(\hat{x}^{0}) \big)$ is a Carath\'{e}odory solution to \eqref{eq:BTobserverDynamics} along the trajectory $x(\cdot)$, starting from the initial condition $b(\hat{x}^{0})$, and if $t \mapsto \xi(t;x(\cdot),\hat{x}^{0})$ is a Carath\'{e}odory solution to \eqref{eq:observerDynamics2}, starting from the initial condition $\hat{x}^{0}$, along the trajectory $x(\cdot)$, then $ \xi(t;x(\cdot),\hat{x}^{0}) = b^{-1}\big(\Psi\big(t;b(x(\cdot)),b(\hat{x}^{0})\big)\big) $ for all $t \in \mathbb{R}_{\geq 0}$.
\end{lem}
\begin{proof}
See \cite[Lemma 2]{SCC.Mahmud.Moad.ea2021}.
\end{proof}
The following section develops a bound on a Lyapunov-like function of the state estimation errors to be utilized in the subsequent stability analysis.

\section{Error bounds for the estimator}\label{Stability Analysis for the state estimation}
\begin{lem}\label{lem3}
Let $V_{se} : \mathbb{R}^{n} \rightarrow \mathbb{R}_{\geq 0}$ be a continuously differentiable candidate Lyapunov function defined as
$ V_{se}(\tilde{s}) \coloneqq \tilde{s}^TP\tilde{s}$ where $P$ is a PD matrix and $P(A-LC)+(A-LC)^{T}P= -\zeta$. Provided $s$, $\hat{s} \in \overline{B}(0,\chi)$ for some $\chi > 0$, the orbital derivative of $V_{se} $ along the trajectories of $\dot{\tilde{s}}$ can be bounded as $\dot{V}_{se} \leq -\left(\lambda_{min}(\zeta)-\varpi_{2}\right) \|\tilde{s}\|^{2}  +\varpi_{1}\|\tilde{s}\|\|s\|  + \varpi_{3} \|\tilde{s}\| \|\tilde{W}_{a}\| +  \varpi_{4} \|\tilde{s}\|$.
\end{lem}
\begin{proof}
See Appendix \ref{appendix:lemma3}.
\end{proof}

The following section develops a novel OF-SMBRL technique for synthesizing  feedback  control  policies.
\section{Safe Model-based Reinforcement Learning}\label{Model-Based Reinforcement Learning}
It is immediate from the Lemma \ref{lem:trajectoryRelation} that if a feedback controller that stabilizes the transformed system in \eqref{eq:BTDynamics} is designed, then the same feedback controller, applied to the original system by inverting the BT also achieves the control objective stated in Section \ref{control object}. In the following, a controller that practically stabilizes \eqref{eq:BTDynamics} is designed as an estimate of a controller that minimizes the infinite horizon cost.
\begin{equation} \label{cost function}
J(u(\cdot)) \coloneqq 	\int_{0}^\infty c(\phi(\tau,s^0,u(\cdot)), u(\tau)) d\tau,
\end{equation}
over the set $\mathcal{U}$ of piecewise continuous functions $t\mapsto u(t)$, 
subject to \eqref{eq:BTDynamics},
where $\phi(\tau, s^0, u (\cdot))$ denotes the trajectory of
(\ref{eq:BTDynamics}), evaluated at time $\tau$, starting from the state $s^0$, and
under the controller $u (\cdot)$. In \eqref{cost function}, $c(s,u) \coloneqq Q'(s) + u^{T}Ru$, with $R \in \mathbb{R}^{m \times m}$ is a symmetric positive definite (PD) matrix. For the optimal value function to be a Lyapunov function for the optimal policy, we are assuming that  $ Q' $ is \text{PD}. A state penalty function $x\mapsto E(x)$, given in the original coordinates, can easily be transformed into an equivalent state penalty $Q'(s) = E(b^{-1}(s))$. Since the barrier function is monotonic and $b(0) = 0$, if $E$ is positive definite, then so is $Q'$. Furthermore, for applications with bounded control inputs, a non-quadratic penalty function similar to \cite[Eq. 17]{SCC.Yang.Ding.ea2020} can be incorporated in \eqref{cost function}. 

Assuming that an optimal controller exists, let the optimal value function, denoted by $V^{*} : \mathbb{R}^{n} \times \mathbb{R}^q  \rightarrow \mathbb{R} $, be defined as
\begin{equation}V^{*}(s) := \min_{u(\cdot)\in \mathcal{U}_{[t,\infty})}\int_{t}^\infty c(\phi(\tau,s,u_{[0,\tau)}(\cdot)), u(\cdot)) d\tau, \label{eq:valuefunction}\end{equation}
where $u_I$ and $\mathcal{U}_I$ are obtained by restricting the domains of $u$ and functions in $\mathcal{U}_I$ to the interval $ I \subseteq \mathbb{R} $, respectively. Assuming that the optimal value function is continuously differentiable, it can be shown to be the unique PD solution of the Hamilton-Jacobi-Bellman (HJB) equation \cite[Theorem 1.5]{SCC.Kamalapurkar.Walters.ea2018}
\begin{equation}\label{HJB} 
\min_{u\in\mathbb{R}^q} \Big(V_{s} \left(F(s)+G(s)u\right) + Q^\prime(s)+ u^{T}Ru\Big) = 0,
\end{equation}
where $\nabla_{\left(\cdot\right)} \coloneqq  \frac{\partial}{\partial \left(\cdot\right)}$, and $V_{\left(\cdot\right)} \coloneqq  \nabla_{\left(\cdot\right)} V$. Furthermore, the optimal controller is given by the feedback policy $u(t) = u^*(\phi(t,s,u_{[0,t)}))$ where $ u^{*}: \mathbb{R}^{n} \rightarrow \mathbb{R}^{m} $ defined as
\begin{equation}\label{eq:optimalcontrol}
    u^{*}(s) := -\frac{1}{2}R^{-1}G(s)^{T}(\nabla_{s}V^{*}(s))^{T}.
\end{equation}

\subsection{Value function approximation}

Considering analytical solutions to the HJB problem are often infeasible to compute, particularly for nonlinear systems, parametric approximation methods are utilized to estimate the value function $V^{*}$ and the optimal policy $u^{*}$.

The optimal value function can be expressed as
\begin{equation} \label{eq:optimalV}
    V^{*}\left(s\right)=W^{T}\sigma\left(s\right)+\epsilon\left(s\right),
\end{equation}
where $W\in\mathbb{R}^{l}$ is an unknown vector of bounded weights, $\sigma:\mathbb{R}^{n}\rightarrow\mathbb{R}^{l}$ is a vector of continuously differentiable nonlinear activation functions \cite[Def. 2.1]{SCC.Sadegh1993} such that $\sigma\left(0\right)=0$ and $\nabla_{s} \sigma \left(0\right)=0$, $l\in\mathbb{N}$ is the number of basis functions, and $\epsilon:\mathbb{R}^{n}\rightarrow\mathbb{R}$ is the reconstruction error.
Using the universal function approximation property \cite[Theorem 1.5]{SCC.Sauvigny2012} of single layer neural networks, it can be deduced that given any compact set\footnote{Note that at this stage, the existence of a compact forward-invariant set that contains trajectories of \eqref{eq:BTDynamics} is not being assumed. The existence of such a set can be established from \cite[Theorem 1a]{SCC.Mahmud.Moad.ea2021}.} $\overline{B}\left(0,\chi\right) \subset\mathbb{R}^{n}$ and a positive constant $\overline{\epsilon}\in\mathbb{R}$, there exists a number of basis functions $l\in\mathbb{N}$, and known positive constants $\bar{W}$ and $\overline{\sigma}$ such that $\left\Vert W\right\Vert \leq\bar{W}$, $\sup_{s\in\overline{B}\left(0,\chi\right)}\left \| \epsilon \left(s\right)\right\| \leq\overline{\epsilon}$, $\sup_{s\in\overline{B}\left(0,\chi\right)}\left\|\nabla_{s}\epsilon\left(s\right)\right\| \leq\overline{\epsilon}$, $\sup_{s\in\overline{B}\left(0,\chi\right)}\left \| \sigma \left(s\right)\right\| \leq\overline{\sigma}$, and $\sup_{s\in\overline{B}\left(0,\chi\right)}\left\|\nabla_{s}\sigma\left(s\right)\right\| \leq\overline{\sigma}$. 

Using (\ref{HJB}), a representation of the optimal controller using the same basis as the optimal value function can be derived as
\begin{equation}\label{eq:optimalu}
    \medmuskip = 0mu
    \thickmuskip = 0mu
    u^{*}\left(s\right)=-\frac{1}{2}R^{-1}G^{T}\left(s\right)\left(\nabla_{s}\sigma^{T}\left(s\right)W+\nabla_{s}\epsilon^{T}\left(s\right)\right).
\end{equation}
Since the ideal weights, $W$, are unknown, an actor-critic technique is utilized in the following to estimate $W$. Let the NN estimates $\hat{V}:\mathbb{R}^{n}\times\mathbb{R}^{l}\to\mathbb{R}$ and $\hat{u}:\mathbb{R}^{n}\times\mathbb{R}^{l}\to\mathbb{R}^{m}$ be defined as
\begin{gather}
\hat{V}\left(\hat{s},\hat{W}_{c}\right)\coloneqq\hat{W}_{c}^{T}\sigma\left(\hat{s}\right),\label{V_app}\\
\hat{u}\left(\hat{s},\hat{W}_{a}\right)\coloneqq-\frac{1}{2}R^{-1}G^{T}\left(\hat{s}\right)\nabla_{\hat{s}}\sigma^{T}\left(\hat{s}\right)\hat{W}_{a},\label{u_app}
\end{gather}
where the critic weights, $\hat{W}_{c}\in\mathbb{R}^{l}$ and actor weights, $\hat{W}_{a}\in\mathbb{R}^{l}$ are estimates of the ideal weights, $W$.
\subsection{Bellman Error}
Substituting (\ref{V_app}) and (\ref{u_app}) into (\ref{HJB}) results in a residual term, $\hat{\delta}: \mathbb{R}^{n} \times \mathbb{R}^{l} \times \mathbb{R}^{l} \rightarrow \mathbb{R}$, referred to as the Bellman error (BE), defined as 
\begin{align} \label{BE1}
    \hat{\delta}(\hat{s},\hat{W}_{c},\hat{W}_{a}) \coloneqq    \hat{V}_{\hat{s}}(\hat{s},\hat{W}_{c}) \left(F(\hat{s}) + G(\hat{s})\hat{u}(\hat{s},\hat{W}_{a})\right) + Q^\prime(\hat{s})+\hat{u}(\hat{s},\hat{W}_{a})^{T}R\hat{u}(\hat{s},\hat{W}_{a}).
\end{align}
To learn the approximation control policy, online RL algorithms traditionally require a persistence of excitation (PE) condition \cite{SCC.Modares.Lewis.ea2013,SCC.Kamalapurkar.Rosenfeld.ea2016,SCC.Kiumarsi.Lewis.ea2014}. It is typically impossible to guarantee PE a priori and verify PE online. While impossible to ensure a priori, by utilizing the model's virtual excitation, stability and convergence of online RL can be established under a PE-like condition that can be validated online (by monitoring the minimum eigenvalue of a matrix in the subsequent Assumption \ref{ass:CLBCADPLearnCond})\cite{SCC.Kamalapurkar.Walters.ea2016}.
Using the system model, the BE can be evaluated at any arbitrary point in the state space. Virtual excitation can then be implemented by selecting a set of state variables $\left\{ r_{k} \mid k=1,\cdots,N\right\} $, where $r_{k}\footnote{In this analysis, the extrapolation state variables $r_{k}$ are assumed to be constant, however, the approach can be extended in a straightforward manner to time-varying extrapolation state variables confined to a compact neighborhood of the origin.}\in \mathbb{R}^{n}$, and evaluating the BE at this set of state variables to yield
\begin{align} \label{BE2}
    \hat{\delta}_{k}(r_{k},\hat{W}_{c},\hat{W}_{a}) \coloneqq \hat{V}_{r_{k}}(r_{k},\hat{W}_{c})\left(F(r_{k}) + G(r_{k})\hat{u}(r_{k},\hat{W}_{a})\right)  + Q^\prime(r_{k})
    + \hat{u}(r_{k},\hat{W}_{a})^{T}R\hat{u}(r_{k},\hat{W}_{a}).
\end{align}
Defining the actor and critic weight estimation errors as $\tilde{W}_{c} \coloneqq W -\hat{W}_{c}$ and  $\tilde{W}_{a} \coloneqq W -\hat{W}_{a}$ and substituting the estimates \eqref{eq:optimalV} and \eqref{eq:optimalu} into (\ref{HJB}), and subtracting from \eqref{BE1}, the BE that can be expressed in terms of the weight estimation errors as\footnote{The dependence of various functions on the state, $s$, is omitted hereafter for brevity whenever it is clear from the context.}
\begin{equation} \label{Analytical BE}
\hat{\delta}=-\omega^{T}\tilde{W}_{c}+\frac{1}{4}\tilde{W}_{a}^{T}G_{\sigma}\tilde{W}_{a}+\Delta,
\end{equation}
where $\Delta\coloneqq\frac{1}{2}W^{T}\nabla_{\hat{s}} \sigma G_{R}\nabla_{\hat{s}} \epsilon^{T}+\frac{1}{4}G_{\epsilon} -\nabla_{\hat{s}}  \epsilon F$,   
 $G_{R}\coloneqq GR^{-1}G^{T}$, $G_{\epsilon}\coloneqq \nabla_{\hat{s}} \epsilon  G_{R} \nabla_{\hat{s}}  \epsilon^{T}$, $G_{\sigma}\coloneqq \nabla_{\hat{s}}  \sigma G R^{-1}G^{T} \nabla_{\hat{s}}  \sigma^{T} $,  and  $\omega \coloneqq  \nabla_{\hat{s}}  \sigma \left(F+G\hat{u}\left(\hat{s},\hat{W}_a\right)\right)$.
Similarly, (\ref{BE2}) implies that
\begin{equation} \label{Approximate BE}
    \hat{\delta}_{k}=-\omega_{k}^{T}\tilde{W}_{c}+\frac{1}{4}\tilde{W}_{a}^{T}G_{\sigma_{k}}\tilde{W}_{a}+\Delta_{k},
\end{equation}
where,
$\Delta_{k} \coloneqq \frac{1}{2}W^{T} \nabla_{r_{k}} \sigma_{k} G_{R_{k}} \nabla_{r_{k}} \epsilon_{k}^{T}+\frac{1}{4}G_{\epsilon_{k}}-\nabla_{r_{k}}\epsilon_{k} F_{k}$, $G_{\epsilon_{k}}\coloneqq \nabla_{r_{k}} \epsilon_{k} G_{R_{k}} \nabla_{r_{k}} \epsilon_{k}^{T}$, \\ $\omega_{k} \coloneqq \nabla_{r_{k}} \sigma_{k}\left(F_k+G_{k}\hat{u}\left(\hat{z}^k,\hat{W}_a\right)\right)$, $G_{\sigma_{k}} \coloneqq \nabla_{r_{k}} \sigma_{k} G_{k} R^{-1} G_{k}^{T} \nabla_{r_{k}} \sigma_{k}^{T}$, $G_{R_{k}}\coloneqq G_{k}R^{-1}G_{k}^{T}$, $F_{k} \coloneqq F(
r_{k})$, $G_{k} \coloneqq G(
r_{k})$, $H_{k} \coloneqq H(
r_{k})$, $\sigma_{k} \coloneqq \sigma (r_{k})$, and $\epsilon_{k} \coloneqq \epsilon(r_{k})$.

Note that $\sup_{s\in\overline{B}\left(0,\chi\right)}\left |\Delta \right| \leq d_{a} \overline{\epsilon}$ and if $r_{k} \in \overline{B}\left(0,\chi\right)$ then $ \left |\Delta_{k} \right| \leq d_{a} \overline{\epsilon}_{k}$, for some constant $d_{a} > 0$.

\subsection{Update laws for Actor and Critic weights}
Using the extrapolated BEs $\hat{\delta}_{k}$ from (\ref{BE2}), the weights are updated according to
\begin{align}
    \dot{\hat{W}}_{c} &=- \frac{k_{c}}{N}\Gamma\sum_{k=1}^{N}\frac{\omega_{k}}{\rho_{k}}\hat\delta_{k},\label{W_c}\\
    \dot{\Gamma} &= \beta\Gamma- \frac{k_{c}}{N}\Gamma\sum_{k=1}^{N}\frac{\omega_{k}\omega_{k}^{T}}{\rho_{k}^{2}}\Gamma,\label{gamma}\\
    \dot{\hat{W}}_{a} \!\!&=\!\! -k_{a_{1}}\!\left(\!\hat{W}_{a}\!-\!\hat{W}_{c}\!\right)\!\!+\!\!\sum_{k=1}^{N}\!\!\frac{k_{c}G_{\sigma_{k}}^{T}\hat{W}_{a}\omega_{k}^{T}}{4N\rho_{k}}\hat{W}_{c}\!-\!k_{a_{2}}\!\hat{W}_{a},\label{W_a}
\end{align}
with $\Gamma\left(t_{0}\right)=\Gamma_{0}$, where $\Gamma:\mathbb{R}_{\geq t_{0}} \to \mathbb{R}^{l\times l}$
is a time-varying least-squares gain matrix, $\rho_{k}\left(t\right)\coloneqq 1+\gamma\omega_{k}^{T}\left(t\right)\omega_{k}\left(t\right)$, $\gamma > 0$ is a constant positive normalization gain,  $\beta > 0 \in \mathbb{R}$ is a constant forgetting factor, and $k_{c},k_{a_{1}},k_{a_{2}} > 0 \in \mathbb{R}$ are constant adaptation gains. 
The control commands sent to the system are then computed using the actor weights as 
\begin{equation}\label{eq:ucontrol}
    u(t)= \hat{u}\left(\hat{s}(t),\hat{W}_{a}(t)\right), \quad t\geq 0.
\end{equation}
The Lyapunov function needed to analyze the closed loop system defined by \eqref{eq:BTDynamics}, \eqref{eq:BTobserverDynamics},  \eqref{W_c}, \eqref{gamma}, and \eqref{W_a}  is constructed using stability properties of \eqref{eq:BTDynamics} under the optimal feedback \eqref{eq:optimalcontrol}. To that end, the following section analyzes the optimal closed-loop system.

Using the assumption that $Q'(s)$ is PD \cite[Theorem 1a]{SCC.Mahmud.Moad.ea2021}, and the converse Lyapunov theorem for asymptotic stability \cite[Theorem 4.17]{SCC.Khalil2002}, the existence of a radially unbounded PD function $ \mathcal{V}:\mathbb{R}^{n}\to\mathbb{R} $ and a PD function $ W:\mathbb{R}^{n}\to\mathbb{R} $ is guaranteed such that
\begin{equation}\label{eq:Converse Lyapunov Function}
\medmuskip=0mu
\thickmuskip=0mu
\thinmuskip=0mu
    \mathcal{V}_{s}\left(s\right)\left(F\left(s\right)+G\left(s\right)u^{*}\left(s\right)\right)\leq-W\left(s\right),
\end{equation}
for all $ s\in\mathbb{R}^{n} $. The functions $ \mathcal{V} $ and $ W $ are utilized in the following section to analyze the stability of the output feedback approximate optimal controller.
\section{Stability Analysis}\label{sec:Analysis}
The following PE-like rank condition is utilized in the stability analysis. 
\begin{assum}
    \label{ass:CLBCADPLearnCond}There exists a constant $\underline{c}_{1} > 0$ such that the set of points $\left\{ r_{k}\in\mathbb{R}^{n}\mid k=1,\hdots,N\right\} $ satisfies
    \begin{equation}
    \underline{c}_{1}I_{L} \leq\inf_{t\in\mathbb{R}_{\geq T}}\left(\frac{1}{N}\sum_{k=1}^{N}\frac{\omega_{k}\left(t\right)\omega_{k}^{T}\left(t\right)}{\rho_{k}^{2}\left(t\right)}\right).\label{eq:CLBCPE2}
    \end{equation}
\end{assum}
Since $\omega_{k}$ is a function of the estimates $\hat{s}$ and $\hat{W}_{a}$,  Assumption \ref{ass:CLBCADPLearnCond} cannot be guaranteed a priori. However, unlike the PE condition, Assumption \ref{ass:CLBCADPLearnCond} can be verified online. Furthermore, since $\lambda_{\min}\left(\sum_{k=1}^{N}\frac{\omega_{k}\left(t\right)\omega_{k}^{T}\left(t\right)}{\rho_{k}^{2}\left(t\right)}\right)$ is non-decreasing in the number of samples, $N$, Assumption \ref{ass:CLBCADPLearnCond} can be met, heuristically, by increasing the number of samples. It is established in \cite[Lemma 1]{SCC.Kamalapurkar.Rosenfeld.ea2016} that under 
Assumption \ref{ass:CLBCADPLearnCond} and provided $\lambda_{\min}\left\{ \Gamma_{0}^{-1}\right\} >0$, the update law in (\ref{gamma}) ensures that the least squares gain matrix satisfies 
\begin{align}\label{eq:OFBADP1Gammabound}
	\underline{\Gamma}I_{L}\leq\Gamma\left(t\right)\leq\overline{\Gamma}I_{L},		\end{align}
$\forall t\in\mathbb{R}_{\geq 0}$ and for some  $\overline{\Gamma},\underline{\Gamma}>0$. Using \eqref{eq:BTDynamics}, the orbital derivative of the PD function $\mathcal{V}$ introduced in \eqref{eq:Converse Lyapunov Function}, along the trajectories of \eqref{eq:BTDynamics}, under the controller $u= \hat{u}\left(\hat{s},\hat{W}_{a}\right)$, is given by $ \dot{\mathcal{V}}\left(s,\tilde{s},\tilde{W}_{a}\right) = \mathcal{V}_{s}\left(s\right)\left(F\left(s\right)+G\left(s\right)\hat{u}\left(\hat{s},\hat{W}_{a}\right)\right)$, where $ \tilde{s} \coloneqq s-\hat{s}$.

Using \eqref{eq:Converse Lyapunov Function} and the facts that $ G $ is bounded, the basis functions $ \sigma $ are bounded, and the value function approximation error $ \epsilon $ and its derivative with respect to $ s, \hat{s} $ are bounded on compact sets, the time-derivative can be bounded as
\begin{equation} \label{eq:Converse Lyapunov}
\dot{\mathcal{V}}\left(s,\tilde{s},\tilde{W}_{a}\right)\leq-W\left(s\right)+\iota_{1}\overline{\epsilon}+\iota_{2}\left\Vert \tilde{s}\right\Vert \left\Vert \tilde{W}_{a}\right\Vert +\iota_{3}\left\Vert \tilde{W}_{a}\right\Vert +\iota_{4}\left\Vert \tilde{s}\right\Vert,
\end{equation}for all $ \hat{W}_{a} \in \mathbb{R}^{l}$, for all $ s \in \overline{B}(0,\chi) $, and for all $ \hat{s} \in \overline{B}(0,\chi) $, where $ \iota_{1},\cdots,\iota_{4} $ are positive constants.

Let $ \Theta\left(\tilde{W}_{c},\tilde{W}_{a},t\right)\coloneqq \frac{1}{2}\tilde{W}_{c}^{T}\Gamma^{-1}\left(t\right)\tilde{W}_{c}+\frac{1}{2}\tilde{W}_{a}^{T}\tilde{W}_{a}$.
 The orbital derivative of $ \Theta $ along the trajectories of \eqref{W_c} - \eqref{W_a} is given by
 \begin{equation} \label{theta}
\dot{\Theta}\left(\tilde{W}_{c},\tilde{W}_{a},t\right)= \tilde{W}_{c}^{T}\Gamma^{-1}\dot{\tilde{W}}_{c}-\frac{1}{2}\tilde{W}_{c}^{T}\Gamma^{-1}\dot{\Gamma}\Gamma^{-1}\tilde{W}_{c}\\+\tilde{W}_{a}^{T}\dot{\tilde{W}}_{a},
\end{equation}
where $\dot{\tilde{W}}_{c}$ = $ -\dot{\hat{W}}_{c}$, and $\dot{\tilde{W}}_{a}$ =  $-\dot{\hat{W}}_{a}$. \\
Provided the extrapolation state variables are selected such that $r_{k} \in \overline{B}(0,\chi)$, $\forall k = \{1,\hdots,N\}$, the orbital derivative in \eqref{theta} can be bounded as 
\begin{multline}\label{eq:thetadot}
       \dot{\Theta}\left(\tilde{W}_{c},\tilde{W}_{a},t\right) \leq -k_{c}\underline{c}\left\Vert \tilde{W}_{c}\right\Vert ^{2}-\left(k_{a1}+k_{a2}\right)\left\Vert \tilde{W}_{a}\right\Vert ^{2}\\+k_{c}\iota_{8}\overline{\epsilon}\left\Vert \tilde{W}_{c}\right\Vert +k_{c}\iota_{5}\left\Vert \tilde{W}_{a}\right\Vert ^{2}+\left(k_{c}\iota_{6}+k_{a1}\right)\left\Vert \tilde{W}_{c}\right\Vert \left\Vert \tilde{W}_{a}\right\Vert \\+\left(k_{c}\iota_{7}+k_{a2}\overline{W}\right)\left\Vert \tilde{W}_{a}\right\Vert, 
\end{multline}
for all $ t\geq0 $, where $ \iota_{5},\hdots,\iota_{8} $ are positive constants that are independent of the learning gains, $ \overline{W} $ denotes an upper bound on the norm of the ideal weights $ W $, and  \\ $ \underline{c} =  \inf_{t \geq 0} \lambda_{\min}\left\{\left(\frac{\beta}{2k_{c}}\Gamma^{-1}\left(t\right)+\frac{1}{2N}\sum_{k=1}^{N}\frac{\omega_{k}\omega_{k}^{T}}{\rho_{k}}\right)\right\}$. Assumption \ref{ass:CLBCADPLearnCond} and \eqref{eq:OFBADP1Gammabound} guarantee that $ \underline{c}>0 $.
From \eqref{eq:derivative_Vse2} we get, 
\begin{align}\label{eq:derivative_Lyapunov_function2f}
\dot{V}_{se} \leq -\left(\lambda_{\min}(\zeta)-\varpi_{2}\right) \|\tilde{s}\|^{2}  +\varpi_{1}\|\tilde{s}\|\|s\|+ \varpi_{3} \|\tilde{s}\| \|\tilde{W}_{a}\| +  \varpi_{4} \|\tilde{s}\|, 
\end{align}
for all $ \hat{W}_{a} \in \mathbb{R}^{l}$, for all $ s \in \overline{B}(0,\chi)$, and for all $\hat{s} \in \overline{B}(0,\chi) $. 
\begin{thm}{\label{thm:Total_Stability_Analysis}}
   
   Provided Assumptions \ref{Assumption-stateestimatorgain} - \ref{ass:CLBCADPLearnCond}, and the hypotheses of Lemma \ref{lem3} hold, the gains are selected large enough to ensure that \eqref{gain_condition} holds and the matrix $M+M^{T}$, is PD, and the weights $\hat{W}_{c}$, $\Gamma$, and $\hat{W}_{a}$ are updated according to \eqref{W_c}, \eqref{gamma}, and \eqref{W_a}, respectively, then the estimation errors $\tilde{W}_{c}$, $\tilde{W}_{a}$, and the trajectories of the transformed system in \eqref{eq:BTDynamics}, under the controller in \eqref{eq:ucontrol}, are locally uniformly ultimately bounded.
\end{thm}
\begin{proof}
See Appendix \ref{appendix:thm:Stability}
\end{proof}

Theorem \ref{thm:Total_Stability_Analysis} and Lemma \ref{lem:trajectoryRelation}-\ref{lem2} shows that the trajectories of the original system in \eqref{eq:Dynamics}, under the controller \eqref{eq:ucontrol}, satisfy the safety constraints and decay to a neighborhood of the origin. It is thus immediate that our developed technique achieves the control objective stated in Section \ref{control object}.










\section{Simulation}\label{Simulation}
To demonstrate the performance of the developed method, the simulation results of the F-16 aircraft longitudinal dynamical system as an example is provided. 

Consider the F-16 aircraft longitudinal dynamical system described by \eqref{eq:Dynamics},
where 
\begin{equation*}\label{sim_dyn}
A = \begin{bmatrix}
-1.01887 & 0.90506 & -0.00215 \\
0.82225 & -1.07741 & -0.17555 \\
0 & 0 & -1 
\end{bmatrix}, 
B = \begin{bmatrix}
0 \\
0  \\
1 
\end{bmatrix}, 
\end{equation*}
\begin{equation*}
C = \begin{bmatrix}
1&0&0
\end{bmatrix}.
\end{equation*}
The state $x$ = $[x_{1};x_{2};x_{3}]$ needs to satisfy the constraints 
$x_{1} \in (\underline{z}_{1},\overline{z}_{1})$, $x_{2} \in (\underline{z}_{2},\overline{z}_{2})$, and $x_{3} \in (\underline{z}_{3},\overline{z}_{3})$ where $\underline{z}_{1}$ =  $\underline{z}_{2}$ = $\underline{z}_{3}$ = $-0.1$, $\overline{z}_{1}$ = $\overline{z}_{2}$ = $\overline{z}_{3}$ = $0.1$.

The objective is to synthesize the action policy to minimize the infinite horizon cost in (\ref{cost function}), with $Q'(s) = s^{T}Qs$ where $Q = 10I_{2}$ and $R = 1$. The basis functions for value function approximation are selected as $\sigma(\hat{s}) = 
[\hat{s}_{1}\hat{s}_{2};\hat{s}_{1}\hat{s}_{3} ;\hat{s}_{2}\hat{s}_{3};\hat{s}_{1}^{2} ; \hat{s}_{2}^{2}; \hat{s}_{3}^{2}]$.

The initial conditions for the state variables, estimated state variables, and the initial guesses for the weights are selected as  $    x(0) = [\alpha_{0};q_{0};\delta{e}_{0}] = [0.045;0;0.0393]$ where $\alpha$ denotes the angle of attack [rad]; $q$ is the pitch rate [rad/sec]; $\delta{e}$ is the elevator deflection angle [rad]. Let, $\hat{x}(0) =  [0.05;0.05;0.05]$, $L = [1;1;1]$ $\Gamma(0)= I_{6}$, $\hat{W}_{a}(0) = [1;1;1;1;1;1]$, $\hat{W}_{c}(0) = [1;1;1;1;1;1]$, $K_{c} = 100$, $K_{a_1} = 100$, $K_{a_2} = 1$, $\beta = 0.1$, and $v = 1$.
The ideal values of the actor and the critic weights are unknown. The simulation uses 125 fixed Bellman error extrapolation points in a 0.08$\times$0.08 $\times$0.08  cube around the origin of the $s-$coordinate system.

\begin{figure}
     \centering
		\includegraphics[width=0.9\columnwidth]{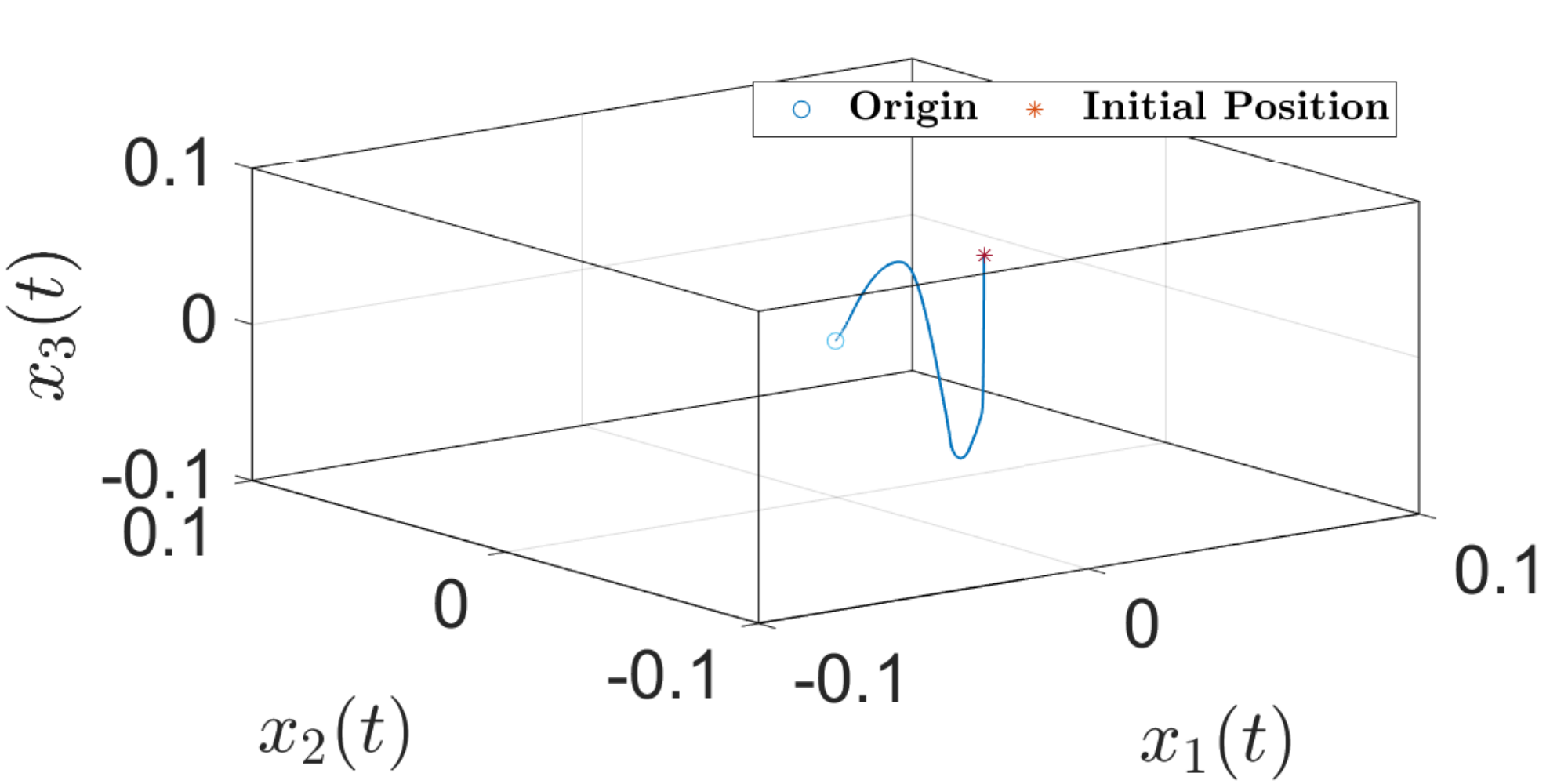}
		\caption{Phase portrait of the F-16 aircraft longitudinal dynamical system's state variables using OF-SMBRL in the original coordinates. The boxed area represents the user-selected safe set.}
		\label{fig:original_state_SE_sim1}
\end{figure}
\begin{figure}
         \centering
		\includegraphics[width=0.8\columnwidth]{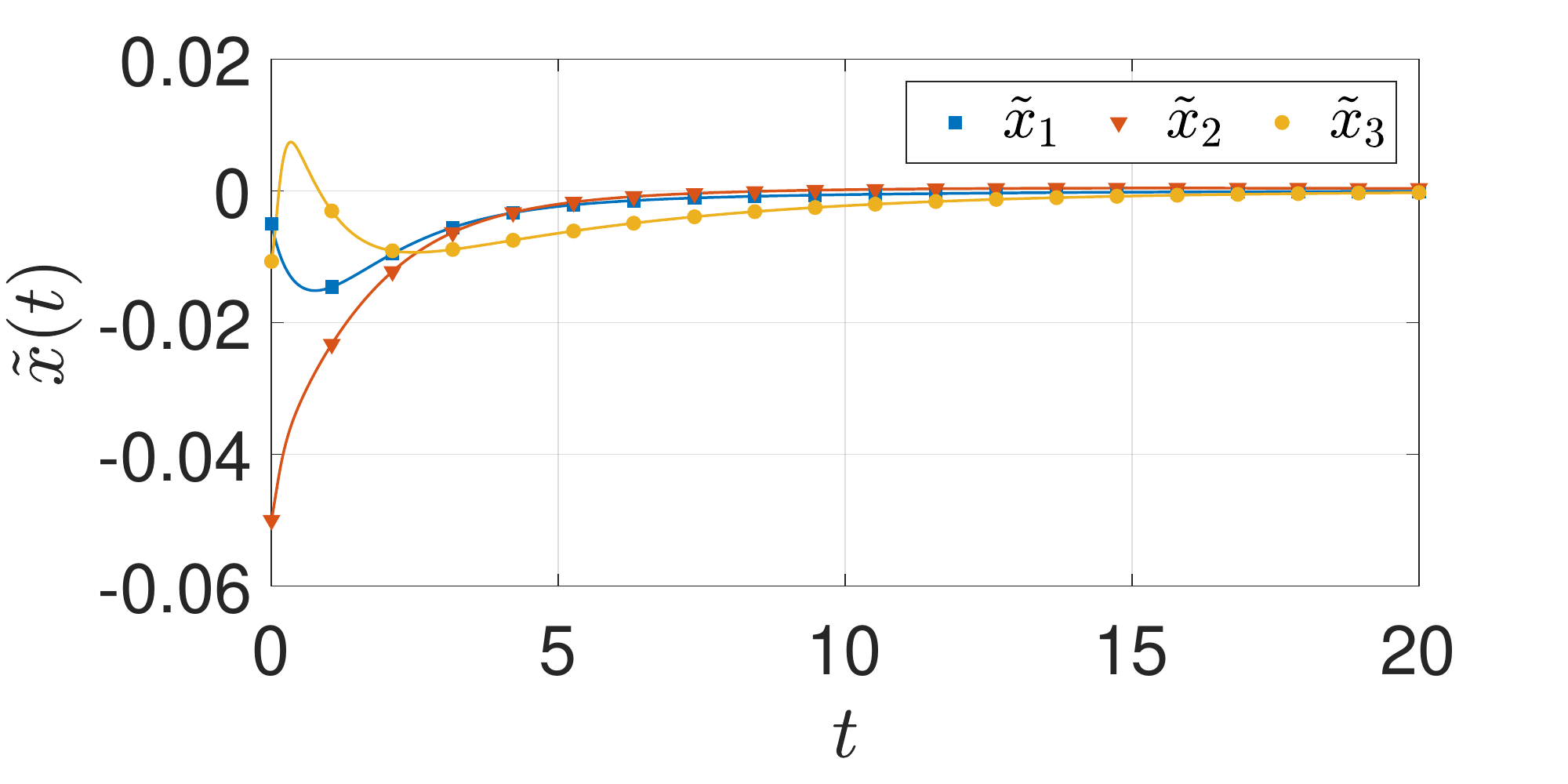}
		\caption{Estimation errors between the original states and the estimated states under selected nominal gains from for the F-16 aircraft longitudinal dynamical system}.
		\label{fig: Estimation error_SE_sim1}
\end{figure}
\begin{figure}
      \centering
		\includegraphics[width=0.8\columnwidth]{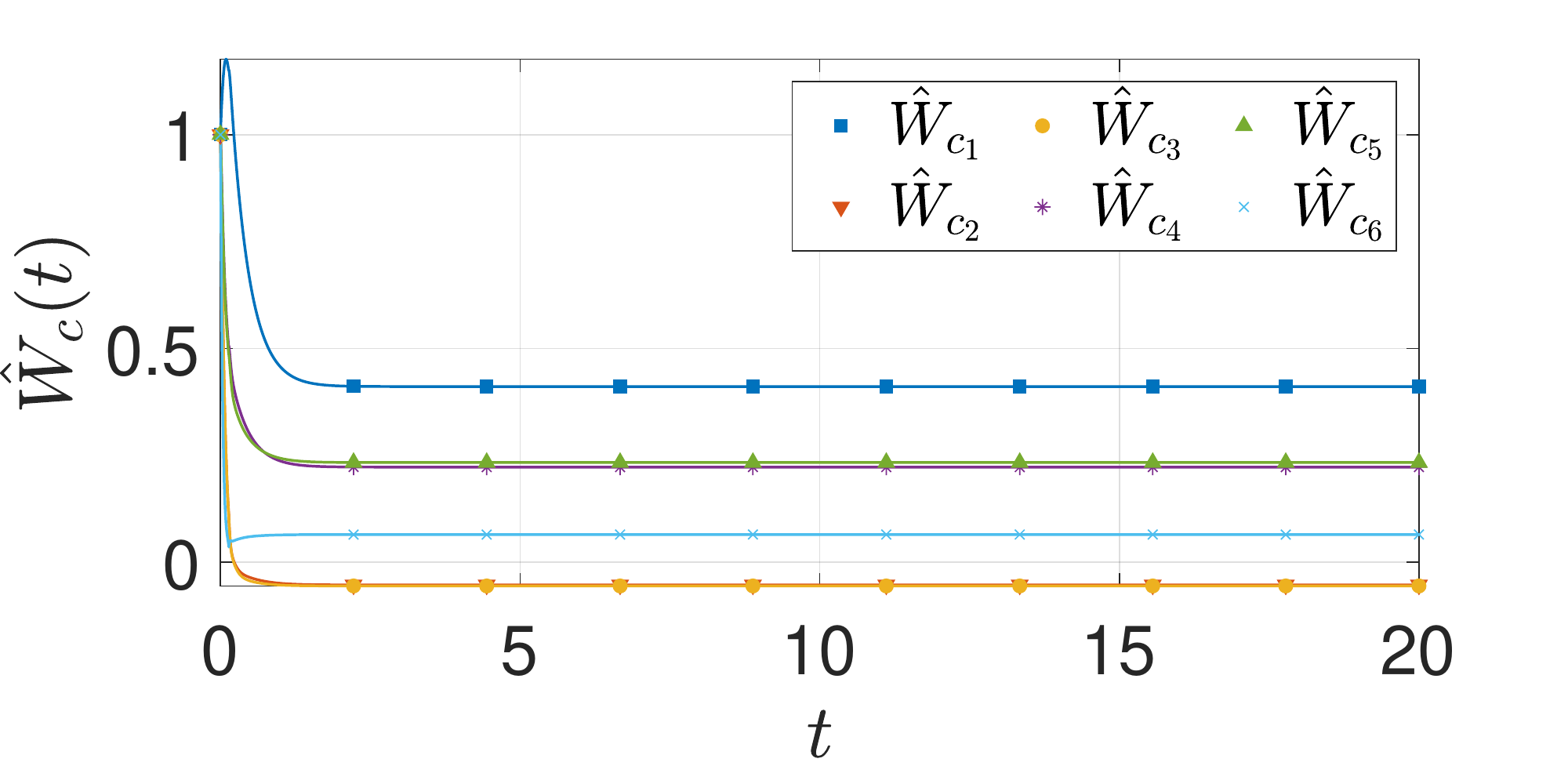}
		\caption{Estimates of the actor and the critic weights under selected nominal gains for the F-16 aircraft longitudinal dynamical system. }
		\label{fig: Weight_EE_SE_sim1}
\end{figure}


Fig.\ref{fig:original_state_SE_sim1} indicates that the system state variables, $x$, remain inside the user-specified safe set when approaching the origin. Fig. \ref{fig: Estimation error_SE_sim1} shows that the state estimation errors also converge to the zero. As observed from the results in Fig. \ref{fig: Weight_EE_SE_sim1}, the unknown weights for both the actor and critic converge to similar values. The plots from the simulation shows the effective performance of the developed method. 

\section{Conclusion}


This paper presents a novel framework that utilizes a new barrier-based adaptive state estimator to yield a safe MBRL based online, approximate optimal controller synthesizing technique for safety-critical linear systems, under output feedback. BT, a transformation method to transform a constrained optimal control problem into an unconstrained optimal control problem, facilitates existing MBRL techniques to obtain safe optimal controllers in the original coordinate. The newly designed Luenberger like state estimator enables the SMBRL framework to provide OF-SMBRL controller that guarantees to keep the state of the original system within the safety bounds. Regulation of the system state variables to a neighborhood of the origin and convergence of the estimated policy to a neighborhood of the optimal policy is established using a Lyapunov-based stability analysis.

While the simulation results are promising, safety violations due to unmodeled uncertainties in the system dynamics and/or the environment are possible. Furthermore, simulation study indicates that the technique is sensitive to initial guesses of the unknown policy and the unknown value function, as predicted by the local stability result. Future research targeted towards these limitations will pave the way for the use of the barrier transformation approach in safety-critical applications such as autonomous driving and manned aviation.

The structure of the sensor matrix used in the state estimator limits the use of sensors to those that directly measure a subset of the state variables. In addition, to address safety, the barrier function can only constrain the state variables of the system to a box. A more generic and adaptive barrier function to relax the limitations is a subject for future research. The barrier transformation approach depends on knowledge of the dynamics of the system to assure safety. If a part of the dynamics is not included in the original model or if there are uncertain parameters in the model, the the relationship  between trajectories of the original dynamics and the transformed system (Lemma \ref{lem:trajectoryRelation}) and the relationship  between the trajectories of the state estimator in the transformed and the original coordinates (Lemma \ref{lem2}) fail to hold. Therefore, more study is required to develop safety assurances under output feedback with parametric uncertainties and/or unmodeled dynamics.


\bibliographystyle{IEEETrans.bst}
\bibliography{scc,sccmaster,scctemp,Extra}

\appendix

\section{Appendix}

\subsection{Proof of Lemma \ref{lem3}}\label{appendix:lemma3}
\begin{manuallemma}{\ref{lem3}}
Let $V_{se} : \mathbb{R}^{n} \rightarrow \mathbb{R}_{\geq 0}$ be a continuously differentiable candidate Lyapunov function defined as
$ V_{se}(\tilde{s}) \coloneqq \tilde{s}^TP\tilde{s}$ where $P$ is a PD matrix and $P(A-LC)+(A-LC)^{T}P= -\zeta$. Provided $s$, $\hat{s} \in \overline{B}(0,\chi)$ for some $\chi > 0$, the orbital derivative of $V_{se} $ along the trajectories of $\dot{\tilde{s}}$ can be bounded as $\dot{V}_{se} \leq -\left(\lambda_{min}(\zeta)-\varpi_{2}\right) \|\tilde{s}\|^{2}  +\varpi_{1}\|\tilde{s}\|\|s\|  + \varpi_{3} \|\tilde{s}\| \|\tilde{W}_{a}\| +  \varpi_{4} \|\tilde{s}\|$.
\end{manuallemma}
\begin{proof}

\begin{table*}[ht]\small
\begin{equation*}
P \coloneqq \begin{bmatrix}
     k_{c}\iota_{8}\overline{\epsilon}\\ \left(k_{c}\iota_{7}+k_{a2}\overline{W}+\iota_{3}\right)  \\ \iota_{4}+\varpi_{4}
\end{bmatrix}^{T},\qquad
M \coloneqq \begin{bmatrix}
    k_{c}\underline{c}&0& 0\\ 
    -\left(k_{c}\iota_{6}+k_{a1}\right)&\left(k_{a1}+k_{a2}-k_{c}\iota_{5}\right)& 0\\
    0&-(\iota_{2}+\varpi_{3})&\lambda_{min} (\zeta)-\frac{1}{2}\varpi_{1}-\varpi_{2}
    \end{bmatrix}^{T}.
    \end{equation*}
    \hrulefill 
\end{table*}


 
The estimator error in the transformed coordinate can be computed as 
\begin{align}
	\dot{\tilde{s}} = F(s)+G(s)\hat{u}-G(\hat{s})\hat{u} -As+A\tilde{s}-LC\tilde{s},
\end{align} 
 
The orbital derivative can be expressed as  
\begin{align}\label{eq:derivative_Lyapunov_function}
\dot{V}_{se} =  \tilde{s}^TP\dot{\tilde{s}}+\dot{\tilde{s}}^TP\tilde{s}, 
\end{align}

Using \eqref{eq:BTobserverDynamics}, we can rewrite \eqref{eq:derivative_Lyapunov_function} as 
\begin{align}
\dot{V}_{se} = \tilde{s}^TP\left( F(s)+G(s)\hat{u}-G(\hat{s})\hat{u} -As+A\tilde{s}-LC\tilde{s}\right)+\left( F(s)+G(s)\hat{u}-G(\hat{s})\hat{u} -As+A\tilde{s}-LC\tilde{s}\right)^{T}P\tilde{s}
\end{align}
which yields 
\begin{align}\label{eq:derivative_Vse}
\dot{V}_{se} = \tilde{s}^TP(A-LC)\tilde{s}+\tilde{s}^TPF(s)+\tilde{s}^TP\tilde{G}(s,\hat{s})\hat{u}-\tilde{s}PAs+\tilde{s}^T(A-LC)^{T}P\tilde{s}+(\tilde{s}^TPF(s))^{T}\nonumber\\+(\tilde{s}^TP\tilde{G}(s,\hat{s})\hat{u})^{T}-(\tilde{s}PAs)^{T}.
\end{align}
We can rewrite \eqref{eq:derivative_Vse} as 
\begin{multline}\label{eq:derivative_Vse1}
\dot{V}_{se} = \tilde{s}^T\bigg(P(A-LC)+(A-LC)^{T}P\bigg)\tilde{s}+\tilde{s}^TPF(s) +\left(\tilde{s}^TPF(s)\right)^{T}-\tilde{s}PAs-(\tilde{s}PAs)^{T} \\   -\tilde{s}^TP\tilde{G}\left(s,\hat{s}\right)\hat{u}\left(s,\tilde{W}_{a}\right)+\tilde{s}^TP\tilde{G}\left(s,\hat{s}\right)\hat{u}\left(s,\tilde{W}_{a}\right)-\tilde{s}^TP\tilde{G}\left(s,\hat{s}\right)\hat{u}\left(\hat{s},\tilde{W}_{a}\right) -\tilde{s}^TP\tilde{G}\left(s,\hat{s}\right)\hat{u}\left(s,W\right)\\+\tilde{s}^TP\tilde{G}\left(s,\hat{s}\right)\hat{u}\left(\hat{s},W\right)+\tilde{s}^TP\tilde{G}\left(s,\hat{s}\right)\hat{u}\left(s,W\right)
-\left(\tilde{s}^TP\tilde{G}\left(s,\hat{s}\right)\hat{u}(s,\tilde{W}_{a})\right)^{T}+\left(\tilde{s}^TP\tilde{G}\left(s,\hat{s}\right)\hat{u}(s,\tilde{W}_{a})\right)^{T}\\-\left(\tilde{s}^TP\tilde{G}\left(s,\hat{s}\right)\hat{u}(\hat{s},\tilde{W}_{a})\right)^{T} -\left(\tilde{s}^TP\tilde{G}\left(s,\hat{s}\right)\hat{u}\left(s,W\right)\right)^{T}+\left(\tilde{s}^TP\tilde{G}\left(s,\hat{s}\right)\hat{u}(\hat{s},W)\right)^{T}+\left(\tilde{s}^TP\tilde{G}\left(s,\hat{s}\right)\hat{u}(s,W)\right)^{T}.
\end{multline}

Let $P(A-LC)+(A-LC)^{T}P= -\zeta$, where $\zeta$ is a PD matrix. With P being a PD matrix, we need to pick a satisfactory gain matrix $L$ so that $(A-LC)$ becomes Hurwitz. 

Considering $\|P\|$ and $\|W\|$ as constants, using the Cauchy-Schwarz inequality and the fact that $F$, $G$ are locally Lipschitz continuous
on $\overline{B}(0, \chi)$, we have 
\begin{align}\label{eq:derivative_Vse2}
\dot{V}_{se} \leq -\lambda_{min}(\zeta)\|\tilde{s}\|^{2}+\varpi_{1}\|\tilde{s}\|\|s\|+ \varpi_{2} \|\tilde{s}\|\|\tilde{s}\| + \varpi_{3} \|\tilde{s}\| \|\tilde{W}_{a}\| +  \varpi_{4} \|\tilde{s}\|
\end{align}
where $\varpi_{1}; \cdots; \varpi_{4}$ as constants. 
From \eqref{eq:derivative_Vse2}, we obtain the desired bound 
\begin{align}\label{eq:derivative_Vse3}
\dot{V}_{se} \leq -\left(\lambda_{min}(\zeta)-\varpi_{2}\right) \|\tilde{s}\|^{2}  +\varpi_{1}\|\tilde{s}\|\|s\|  + \varpi_{3} \|\tilde{s}\| \|\tilde{W}_{a}\| +  \varpi_{4} \|\tilde{s}\|. \end{align}

\end{proof}
\subsection{Proof of Theorem \ref{thm:Total_Stability_Analysis}}\label{appendix:thm:Stability}
\begin{manualtheorem}{\ref{thm:Total_Stability_Analysis}}
    Provided Assumptions \ref{Assumption-stateestimatorgain} - \ref{ass:CLBCADPLearnCond}, and the hypotheses of Lemma \ref{lem3} hold, the gains are selected large enough to ensure that \eqref{gain_condition} holds and the matrix $M+M^{T}$, is PD, and the weights $\hat{W}_{c}$, $\Gamma$, and $\hat{W}_{a}$ are updated according to \eqref{W_c}, \eqref{gamma}, and \eqref{W_a}, respectively, then the estimation errors $\tilde{W}_{c}$, $\tilde{W}_{a}$, and the trajectories of the transformed system in \eqref{eq:BTDynamics}, under the controller in \eqref{eq:ucontrol}, are locally uniformly ultimately bounded.
\end{manualtheorem}
\begin{proof}

The candidate Lyapunov function for the closed-loop system is selected as
\begin{equation}\label{eq:candidateLyapunovfunction}
V_{L}\left(Z,t\right)\coloneqq \mathcal{V}\left(s\right)+\Theta\left(\tilde{W}_{c},\tilde{W}_{a},t\right)+V_{se}\left(\tilde{s}\right), 
\end{equation}
where $ Z\coloneqq\begin{bmatrix}
s^{T}&\tilde{W}_{c}^{T}&\tilde{W}_{a}^{T}&\tilde{s}^{T}
\end{bmatrix}^{T}.$

Let $\mathcal{C} \subset \mathbb{R}^{2n}$ be a compact set defined as
\[
    \mathcal{C} \coloneqq \left\{(s,\tilde{s}) \in \mathbb{R}^{2n} \mid \begin{gathered}  \|s\|+\|\tilde{s}\|\end{gathered}\right\}.
\]

Whenever, $(s,\tilde{s}) \in \mathcal{C}$, it can be concluded that $s,\hat{s} \in \overline{B}(0,\chi)$. As a result, \eqref{eq:Converse Lyapunov}, \eqref{eq:thetadot}, and \eqref{eq:derivative_Vse2} imply that whenever $Z \in \mathcal{C} \times \mathbb{R}^{2l}$, the orbital derivative of the candidate Lyapunov function along the trajectories of \eqref{eq:BTDynamics}, \eqref{eq:BTobserverDynamics}, \eqref{W_c}, \eqref{gamma}, \eqref{W_a}, under the controller \eqref{eq:ucontrol}, can be bounded as
 \begin{align*}
\dot{V}_{L}\left(Z,t\right) \leq -W\left(s\right)+\iota_{1}\overline{\epsilon}+\iota_{2}\left\Vert \tilde{s}\right\Vert \left\Vert \tilde{W}_{a}\right\Vert +\iota_{3}\left\Vert \tilde{W}_{a}\right\Vert +\iota_{4}\left\Vert \tilde{s}\right\Vert-k_{c}\underline{c}\left\Vert \tilde{W}_{c}\right\Vert ^{2}-\left(k_{a1}+k_{a2}\right)\left\Vert \tilde{W}_{a}\right\Vert ^{2}+k_{c}\iota_{8}\overline{\epsilon}\left\Vert \tilde{W}_{c}\right\Vert \\+k_{c}\iota_{5}\left\Vert \tilde{W}_{a}\right\Vert ^{2}+\left(k_{c}\iota_{6}+k_{a1}\right)\left\Vert \tilde{W}_{c}\right\Vert \left\Vert \tilde{W}_{a}\right\Vert +\left(k_{c}\iota_{7}+k_{a2}\overline{W}\right)\left\Vert \tilde{W}_{a}\right\Vert-\big(\lambda_{min} (\zeta)-\varpi_{2}\big)\|\tilde{s}\|^{2} \\+\big(\varpi_{1}\big)\|\tilde{s}\|\|s\| + \varpi_{3} \|\tilde{s}\| \|\tilde{W}_{a}\| +  \varpi_{4} \|\tilde{s}\|,
\end{align*}

which can be re-expressed as,
\begin{multline*}
\dot{V}_{L}\left(Z,t\right) \leq -W\left(s\right)+\iota_{1}\overline{\epsilon}+(\iota_{2}+\varpi_{3})\left\Vert \tilde{s}\right\Vert \left\Vert \tilde{W}_{a}\right\Vert -k_{c}\underline{c}\left\Vert \tilde{W}_{c}\right\Vert ^{2}-\left(k_{a1}+k_{a2}-k_{c}\iota_{5}\right)\left\Vert \tilde{W}_{a}\right\Vert ^{2}+k_{c}\iota_{8}\overline{\epsilon}\left\Vert \tilde{W}_{c}\right\Vert \\+\left(k_{c}\iota_{6}+k_{a1}\right)\left\Vert \tilde{W}_{c}\right\Vert \left\Vert \tilde{W}_{a}\right\Vert +\left(k_{c}\iota_{7}+k_{a2}\overline{W}+\iota_{3}\right)\left\Vert \tilde{W}_{a}\right\Vert +(\iota_{4}+\varpi_{4})\left\Vert \tilde{s}\right\Vert-\big(\lambda_{min} (\zeta)-\varpi_{2}\big)\|\tilde{s}\|^{2} +\varpi_{1}\|\tilde{s}\|\|s\|,
\end{multline*}
which yields using Young's inequality
\begin{multline*}
\dot{V}_{L}\left(Z,t\right) \leq -W\left(s\right)+\frac{1}{2}(\varpi_{1})\|s\|^{2}+\iota_{1}\overline{\epsilon}+(\iota_{2}+\varpi_{3})\left\Vert \tilde{s}\right\Vert \left\Vert \tilde{W}_{a}\right\Vert  +(\iota_{4}+\varpi_{4})\left\Vert \tilde{s}\right\Vert-k_{c}\underline{c}\left\Vert \tilde{W}_{c}\right\Vert ^{2}\\-\left(k_{a1}+k_{a2}-k_{c}\iota_{5}\right)\left\Vert \tilde{W}_{a}\right\Vert ^{2}+k_{c}\iota_{8}\overline{\epsilon}\left\Vert \tilde{W}_{c}\right\Vert +\left(k_{c}\iota_{6}+k_{a1}\right)\left\Vert \tilde{W}_{c}\right\Vert \left\Vert \tilde{W}_{a}\right\Vert +\left(k_{c}\iota_{7}+k_{a2}\overline{W}+\iota_{3}\right)\left\Vert \tilde{W}_{a}\right\Vert\\-\big(\lambda_{min} (\zeta)-\frac{1}{2}\varpi_{1}-\varpi_{2}\big)\|\tilde{s}\|^{2},
\end{multline*}


where $ z\coloneqq \begin{bmatrix}\left\Vert \tilde{W}_{c}\right\Vert  & \left\Vert \tilde{W}_{a}\right\Vert  & \left\Vert \tilde{s}\right\Vert  \end{bmatrix}^{T}$.
Provided the matrix $ M+M^{T}$  is PD,
\begin{equation*}
\dot{V}_{L}\left(Z,t\right)\leq-W\left(s\right)+\frac{1}{2}(\varpi_{1})\|s\|^{2}-\underline{M}\left\Vert z\right\Vert ^{2}+\overline{P}\left\Vert z\right\Vert +\iota_{1}\overline{\epsilon},
\end{equation*}where $ \underline{M} \coloneqq \lambda_{\min}\left\{\frac{M+M^{T}}{2}\right\}$, $\overline{P}= \|P\|_{\infty}$. Letting $ \underline{M}\eqqcolon\underline{M}_{1}+\underline{M}_{2}$, and letting $\mathcal{W}:\mathbb{R}^{2n+2l} \to \mathbb{R} $ be defined as $ \mathcal{W}\left(Z\right)=-W\left(s\right)+\frac{1}{2}(\varpi_{1})\|s\|^{2}-\underline{M}_{1}\left\Vert z\right\Vert^{2}$, with $W(s)> \frac{1}{2}(\varpi_{1})\|s\|^{2}$, the orbital derivative can be bounded as 
\begin{equation}
\dot{V}_{L}\left(Z,t\right)\leq-\mathcal{W}\left(Z\right),\label{eq:OFBADPVDotBound}
\end{equation}
$\forall \left\Vert Z\right\Vert >\frac{1}{2}\left(\frac{\overline{P}}{\underline{M}_{2}}+\sqrt{\frac{\overline{P}^{2}}{\underline{M}_{2}^{2}}+\frac{\iota_{1}^{2}\overline{\epsilon}^{2}}{\underline{M}_{2}^{2}}}\right)\eqqcolon \mu$, $\forall Z\in\overline{B}\left(0,\bar {\chi}\right)$, for all $ t\geq 0 $, and some $\bar{\chi}$ such that $\overline{B}(0,\bar{\chi}) \subseteq \mathcal{C} \times \mathbb{R}^{2l}$. 

Using the bound in \eqref{eq:OFBADP1Gammabound} and the fact that the converse Lyapunov function is guaranteed to be time-independent, radially unbounded, and PD, \cite[Lemma 4.3]{SCC.Khalil2002} can be invoked to conclude that \begin{equation}
\underline{v}\left(\left\Vert Z\right\Vert \right)\leq V_{L}\left(Z,t\right)\leq\overline{v}\left(\left\Vert Z\right\Vert \right),\label{eq:OFBADPVBound}
\end{equation}
for all $t \in \mathbb{R}_{\geq 0}$ and for all $Z\in\mathbb{R}^{2n+2l}$, where $\underline{v},\overline{v}:\mathbb{R}_{\geq 0}\rightarrow\mathbb{R}_{\geq 0}$ are class $\mathcal{K}$ functions.

Provided the learning gains, the domain radii $ \chi \ \text{and} \ \bar{\chi} $, and the basis functions for function approximation are selected such that $ M+M^{T} $ is PD and 
\begin{equation} \label{gain_condition}
    \mu<\overline{v}^{-1}\left(\underline{v}\left(0,\bar {\chi}\right)\right),
\end{equation}
then \cite[Theorem 4.18]{SCC.Khalil2002} can be invoked to conclude that Z is locally uniformly ultimately bounded. Since the estimates $\hat{W}_{a}$ approximate the ideal weights $ W $, the policy $ \hat{u} $ approximates the optimal policy $ u^{*} $.
\end{proof}
\end{document}